\documentclass{ws-ijfcs}

\usepackage[english]{babel}
\usepackage{times}
\usepackage{epsfig}
\usepackage{xspace}
\usepackage{amsmath, amssymb}
\usepackage{algorithm}
\usepackage[noend]{algorithmic}
\algsetup{indent=1.3em}
\usepackage{tikz}
\usepackage{varioref}
\usepackage{multirow}
\usepackage{bold-extra}
\usepackage{mathrsfs}
\usepackage{graphicx}
\usepackage{setspace}
\usepackage{slashbox}
\usepackage{placeins}

\newcommand{\N}{{\mathbb N}}
\newcommand{\R}{{\mathbb R}}

\newcommand{\T}{\mathbb{T}}
\newcommand{\lifetime}{{\cal T}}

\newcommand{\J}{{\cal J}}

\newcommand{\G}{{\cal G}}

\newcommand{\tdbroadcast}{{\textsc{TDB}}\xspace}
\newcommand{\foremost}{{\textsc{TDB}}\ensuremath{[foremost]}\xspace}
\newcommand{\fastest}{{\textsc{TDB}}\ensuremath{[fastest]}\xspace}
\newcommand{\shortest}{{\textsc{TDB}}\ensuremath{[shortest]}\xspace}
\newcommand{\RG}{\ensuremath{{\cal R}}\xspace}
\newcommand{\BG}{\ensuremath{{\cal B}}\xspace}
\newcommand{\PG}{\ensuremath{{\cal P}}\xspace}
\newcommand{\PP}{\ensuremath{\mathscr{P}}\xspace}

\newtheorem{prop}[theorem]{Property}
 
\newcounter{smallitemizec}

\xdefinecolor{red}{rgb}{0.6,0,0}
\xdefinecolor{green}{rgb}{0,0.5,0}

\DeclareFixedFont{\petit}{\encodingdefault}%
{\familydefault}{\seriesdefault}{\shapedefault}{10pt}

\makeatletter
\def\ps@pprintTitle{}

\newcommand{\mylno}{%
\ifthenelse{\equal{\arabic{ALC@rem}}{0}}
{{\footnotesize \arabic{ALC@line}:~~\quad}}{}%
}
\newcommand{\codetitle}[1]{\medskip\STATE \hspace{-12pt}\underline{#1:}\vspace{2pt}}

\newenvironment{code}[1][0]%
{\begin{algorithmic}[#1]\let\ALC@lno=\mylno}
{\end{algorithmic}}
\makeatother

\newenvironment{proof-idea}{\noindent{\it Idea of the proof.}\hspace*{.5em}}{\bigskip}
\newenvironment{proof-sketch}{\noindent{\it Proof sketch.}\hspace*{.5em}}{\bigskip}

\begin{document}

\title{\large Shortest, Fastest, and Foremost Broadcast in Dynamic Networks
 \thanks{Preliminary results were presented at the  6th IFIP International Conference on Theoretical Computer Science. This work has been supported in part by the ARC (Australia), NSERC (Canada) and Dr. Flocchini's University Research Chair.}\vspace{-10pt}
}

\author{
  {\normalsize
 Arnaud Casteigts$^1$,
 Paola Flocchini$^2$,
 Bernard Mans$^3$, and
 Nicola Santoro$^4$\vspace{10pt}\\
}
 $^1$ University of Bordeaux, France\\
 $^2$ University of Ottawa, Canada\\
$^3$ Macquarie University, Sydney, Australia\\
$^4$ Carleton University, Ottawa, Canada\\
}

\maketitle

\begin{abstract}
Highly dynamic networks rarely offer end-to-end connectivity at a given time. 
Yet, connectivity in these networks can be established over time and space, based on temporal analogues of multi-hop paths (also called {\em journeys}). 
Attempting to optimize the selection of the journeys in these networks naturally leads to the study of three cases: shortest (minimum hop), fastest (minimum duration), and foremost (earliest arrival) journeys. 
Efficient centralized algorithms exists to compute all cases, when the full knowledge of the network evolution is given.

In this paper, we study the {\em distributed} counterparts of these problems, i.e. shortest, fastest, and foremost broadcast with termination detection (TDB), with minimal knowledge on the topology.
 We show that the feasibility of each of these problems requires distinct features on the evolution, through identifying three classes of dynamic graphs wherein the problems become gradually feasible: 
graphs in which the re-appearance of edges is {\em recurrent} (class \RG), {\em bounded-recurrent} (\BG), or {\em periodic} (\PG), together with specific knowledge that are respectively $n$ (the number of nodes), $\Delta$ (a bound on the recurrence time), and $p$ (the period). In these classes it is not required that all pairs of nodes get in contact -- only that the overall  {\em footprint} of the graph is connected over time.

Our results, together with the strict inclusion between $\PG$, $\BG$, and $\RG$, implies a feasibility order among the three variants of the problem, i.e. $\foremost$ requires weaker assumptions on the topology dynamics than $\shortest$, which itself requires less than $\fastest$. Reversely, these differences in feasibility imply that the computational powers of $\RG_n$, $\BG_\Delta$, and $\PG_p$ also form a strict hierarchy.

\end{abstract}


\section{Introduction}


{\em Dynamic networks} are widely addressed in distributed computing. Contexts of interest are as varied as fault-tolerance, interaction scheduling, dynamic membership, planned mobility, or unpredictable mobility. 
The recent emergence of scenarios where entities are truly mobile and can communicate without infrastructure (e.g. vehicles, satellites, robots, or pedestrian smartphones) brought to the fore the most versatile of these environments. In these {\em highly} dynamic networks, changes are not anomalies but rather integral part of the nature of the system. 


The need to categorize and understand highly dynamic networks led the engineering community to design a variety of {\em mobility models}, each of which captures a particular context by means of rules that determine how the nodes move and communicate (see e.g.~\cite{HFB09}). A popular example includes the well-known {\em random waypoint} model~\cite{BRS03}. The main interest of these models is to be able to reproduce experiments and compare different solutions on a relatively fair basis, thereby providing a common ground to solve practical challenges in highly dynamic networks, e.g.~routing and broadcasting~\cite{BGJL06,GK07,JMR10,JFP04,LW09b,Zha06,ZAZ04}.


In the same way as mobility models enable to federate practical investigations in highly dynamic networks, {\em logical properties} on the graph dynamics, that is, {\em classes of dynamic graphs}, have the potential to guide a more formal exploration of their analytical aspects. A number of special classes were recently identified, for instance graphs in which the nodes interact infinitely often ({\it e.g.,} uniform random scheduler for {\it population protocols}~\cite{AAD+06,AAER07,CMS09}); graphs whose dynamics is unrestricted but remains connected at any instant~\cite{DPR+13,OW05}; graphs in which there exists a stable connected spanning subgraph in any T-time window ({\it a.k.a. T-interval connectivity})~\cite{HK12,KLO10}; graphs whose edges appear or disappear with given probabilities~\cite{BCF09,CleMMPS08,CMPS11,PSSS11}; graphs that have a stable root component~\cite{BRS12}; graphs whose schedule is periodic~\cite{CFMS14,FKMS12a,FMS13,IW11} or guarantees minimal reachability properties~\cite{CCF09}. These classes (among others) were 
characterized within a common formal framework  and organized into a hierarchy in~\cite{CFQS12}.


In this paper we are interested in studying specific relationship between some of these classes, namely three subclasses of those networks called {\em  delay-tolerant networks} (DTNs), in which instant connectivity is never guaranteed, but still connectivity can be achieved over time and space (see e.g.~\cite{AE84}). These classes are:
\begin{itemize}
  \item Class \RG of all graphs whose edges cannot disappear forever ({\em recurrent edges}). That is, if an edge appears once and disappears, then it will eventually re-appear at some unknown (but finite) date. 
  It is  not required that all pairs of nodes share an edge, but only that the {\em footprint} of all edges forms a connected graph (otherwise, even temporal connectivity is not guaranteed). This class corresponds to Class~6 in~\cite{CFQS12}.
\item Class  $\cal{B}$ (for {\em bounded-recurrent} edges) consisting of those graphs with recurrent edges in which the recurrence time cannot exceed a given duration $\Delta$. And again, the footprint is connected. This class corresponds to Class~7 in~\cite{CFQS12}.
\item Class $\cal{P}$ (for {\em periodic} edges) consisting of those graphs in which all topological events (appearance or disappearance) repeat identically modulo some period $p$. And again, the footprint is connected. This class corresponds to Class~8 in~\cite{CFQS12}.
\end{itemize}


As far as {\em inclusion} is concerned, 
it clearly holds that $\cal{P} \subset \cal{B} \subset \cal{R}$, 
but what about the {\em computational relationship} between these classes? 
Considering different types of knowledge, namely the number $n$ of nodes in the network, a bound $\Delta$ on the recurrence time, and (any multiple of) the period $p$, we look at the relationship between
$\PP(\RG_n)$, $\PP(\BG_\Delta)$, and $\PP(\PG_p)$, 
where $\PP({\cal C}_K)$ is the set of problems one can solve in class ${\cal C}$ with knowledge $K$. 

The investigation is carried out by studying a fundamental problem in distributed computing: 
{\em broadcast} with termination detection at the emitter (or \tdbroadcast);
this problem  is also known in the literature  with
different names (Echo, Propagation of Information with Feedback, etc.) or
context (synchronizers, etc.).
It can have at least three distinct definitions in highly dynamic networks:
\foremost, in which the date of delivery is minimized at every node;
\shortest, where the number of hops used by the broadcast is minimized relative to every node; and 
\fastest, where the overall duration of the broadcast is minimized (however late the departure be).
These three metrics were considered
  in the seminal work by Bui-Xuan, Ferreira, and Jarry \cite{BFJ03} 
where the authors solved the offline problem  of computing   all shortest, fastest, and foremost journeys from a given node, 
 given a complete schedule of the network.

\subsection*{Main contributions}

In this paper we examine the feasibility and reusability of the solution (and to some extent, the complexity)
 of \foremost, \shortest, and \fastest in \RG, \BG, \PG with knowledge $\emptyset$, $n$,  or $\Delta$.
  We additionally draw some observations from existing  results in~\PG with knowledge $p$~\cite{CFMS14}, that complete our general picture of feasibility and reusability of broadcast in the three classes. 
 Here is a short summary of some of the contributionss.

\paragraph{Feasibility:} 
We first show that none of these problems are solvable  in any of the classes 
unless additional knowledge is considered. 
We then establish 
several results, both positive and negative, on the feasibility  
of \foremost, \shortest, and \fastest in \RG, \BG, and \PG. In particular, we 
prove
that knowing $n$ 
makes it possible to solve \foremost in \RG, 
but this is not sufficient to solve \shortest nor \fastest, even in \BG.
\shortest becomes in turn feasible in \BG if $\Delta$ is known, 
but this is not sufficient to solve \fastest;
this later problem becomes solvable in \PG knowing $p$~\cite{CFMS14}. 
These results allow us to show that 
\begin{equation}
\PP(\RG_n)\subsetneq\PP(\BG_\Delta)\subsetneq\PP(\PG_p)
\end{equation}

\noindent 
where    $\PP({\cal C}_K)$ denotes the set of problems solvable in every ${\cal G}\in{\cal C}$ with knowledge $K$.
In other words, the computational relationships between these three contexts form a {\em strict} hierarchy.


In the universe  
${\cal U} = \{ \RG_n, \BG_n, \BG_\Delta,\BG_{\{n,\Delta\}},  \PG_n,   \PG_\Delta,   \PG_{\{n,\Delta\}},  \PG_p \}$ 
of  the classes of dynamic networks with  knowledge considered here, let
 $P_1 \preceq  P_2$ denote the fact that $P_1$ is ``no more difficult" than  $P_2$, that is,
  if $P_2$ is solvable in ${\cal G}_K\in{\cal U}$ so is $P_1$; 
 and let $P_1 \prec  P_2$ denote the fact that $P_1$ is ``less difficult" than  $P_2$, that is, 
 $P_1 \preceq  P_2$ and there exists ${\cal G}_K\in{\cal U}$ in which $P_1$ is solvable but $P_2$ is not.
Our results show the existence of a {\em strict} hierarchy between these problems  with respect to {\em feasibility}:
 
 \begin{equation}
  \label{eq:relationship}
  \foremost \prec  \shortest \prec    \fastest
\end{equation}

\noindent   These  results   are summarized in Table \ref{tab:results-feasibility}.

 \newcommand{\ov}[1]{$\overline{\mbox{#1}}$}
 \def\arraystretch{1.3}
 
\begin{table}[h]
\label{tab:results-feasibility}
\centering
\begin{tabular}{|c|@{}c@{}|@{\,}c@{}|@{}c@{}|c@{}|}
     \hline
     \backslashbox{Class\hspace*{-30pt}}{Knowledge\hspace*{-5pt}}&\large $\emptyset$&\large $n$&\large $\Delta$&\large $p$\\\hline
     $\cal{R}$&\multirow{3}{*}{
        \begin{minipage}[c]{1.7cm}
         $\neg$ Foremost\\$\neg$ Shortest\\$\neg$ Fastest
       \end{minipage}}&\multirow{3}{*}{\begin{minipage}[c]{1.6cm}
         Foremost\\$\neg$ Shortest\\$\neg$ Fastest
       \end{minipage}}&$n/a$&\multirow{2}{*}{$n/a$}\\\cline{1-1}\cline{4-4}
     $\cal{B}$&&&\multirow{2}{*}{
       \begin{minipage}[c]{2.7cm}
         Foremost, Shortest\\
         $\neg$ Fastest
       \end{minipage}}&\\\cline{1-1}\cline{5-5}
     $\cal{P}$&&&&\hspace{-4pt}Foremost, Shortest, Fastest$^*$\\\hline
   \end{tabular}
  \def\arraystretch{1}
 \caption{
Feasibility of broadcast with 
termination detection, depending on the class of dynamic networks ($\cal R,B,P$) 
and of  knowledge ($\emptyset,n,\Delta,p$).  Unfeasibility is denoted by  $``\neg"$.
The feasibility of Fastest in ${\cal P}$ with knowledge $p$ ($ ^*$ in table) is from~\protect\cite{CFMS14}.}
\end{table}

\paragraph{Reusability:} 
Regarding the possibility to reuse a solution, that is, a same broadcast tree, over several broadcasts,
we establish several results; in particular we show the intriguing fact that reusability in \shortest is easier than
that of \foremost. Precisely, when \shortest becomes feasible in \BG,
it enables at once reusability of the broadcast trees, 
whereas \foremost, although it was already feasible in \RG, 
does not enable reusability until in \PG~\cite{CFMS14}.
 
For reusability, let    relations $\leq $ and $<$  
 be the analogous   of $\preceq$ and $\prec$, respectively; and let 
  $P_1 \equiv  P_2$ denote that both  $P_1 \leq  P_2$  and $P_2 \leq P_1$.
Our results imply that:

\begin{equation}
  \shortest < \foremost \equiv  \fastest 
\end{equation}


\paragraph{Complexity:} Although complexity is not the main focus here, 
we characterize the time complexity and message complexity
of our algorithms and observe some interesting facts. 
For instance, the message complexity of our algorithm for \foremost
is lower knowing $\Delta$ than knowing $n$, 
and even lower if both are known. These results are summarized in Table~2. Note that \tdbroadcast involves two processes: the actual dissemination of {\em information messages}, and the exchange of typically smaller {\em control messages} (e.g. for termination detection), both of which are separately analyzed. 
Regarding time complexity, we observe that for all algorithms but those which terminate implicitly, the termination detection phase takes the same order of time as the dissemination phase. Thus, Table~2 does not distinguish both phases.

 \begin{table}[h]
   \begin{center}
\label{tab:results-complexity}
 \begin{tabular}{|@{\,}c@{\,}|@{~}c@{~}|@{~}c@{~}|@{~}c@{~}|@{~}c@{~}|@{\,}c@{~}|@{~}|@{~}c@{~}|@{~}c@{~}|}
\hline
Metric&Class&Knowl.&Time&Info. msgs&Control msgs&Info. msgs&Control msgs\\
~&~&~&~&($1^{st}$ run)&($1^{st}$ run)&(next runs)&(next runs)\\\hline
  Foremost&$\cal R$&     $n$&\small unbounded&$O(m)$&$O(n^2)$&$O(m)$ &$O(n)$\\
         ~&$\cal B$&     $n$&$O(n\Delta)$&$O(m)$&$O(n^2)$&$O(m)$&$O(n)$\\
         ~&       ~&        $\Delta$&$O(n\Delta)$&$O(m)$&$O(n)$&$O(m)$&0\\
         ~&       ~&   $n \& \Delta$&$O(n\Delta)^*$&$O(m)$&0&$O(m)$&0\\
         \hline
  Shortest&$\cal B$&        $\Delta$&$O(n\Delta)$&$O(m)$&$O(n)$: $2n-2$&$O(n)$&0\\
         \multirow{2}{*}{{\it either of} {\Large \{}}&       ~&   $n \& \Delta$&$O(n\Delta)$&$O(m)$&$O(n)$: $n-1$&$O(n)$&0\\
         &       ~&   $n \& \Delta$&$O(n\Delta)^*$&$O(m)$&0&$O(m)$&0\\
  \hline 
\end{tabular}
\end{center}
\caption{Complexity of \tdbroadcast in different classes of dynamic networks with associated knowledge. {\it (The $^*$ indicates that the emitter terminates implicitely, even in the first run.)}
}
\end{table}






 \section{Model and Basic Properties}

\subsection{Definitions and Terminology}
\label{sec:defs}

Consider a system composed of a finite set of $n$ entities $V$ (or nodes) that interact with each other over a (possibly infinite) time span 
$\lifetime \subseteq \T$ called {\em lifetime} of the system, where $\T$ is the temporal domain (typically,  $\N$ or $ \R^+$ for discrete and continuous-time systems, respectively). In this paper we consider a continuous-time setting with $\T=\R^+$.

Following~\cite{CFQS12}, we describe the network as a {\em
  time-varying graph} ({\em TVG}, for short)
$\G=(V,E,\lifetime,\rho,\zeta)$, where $E \subseteq V\times V$ is a
set of $m$ (possibly {\em intermittent}) undirected edges such that
$(u,v) \in E \Leftrightarrow$ $u$ and $v$ have at least one contact
over $\lifetime$; $\rho : E \times \lifetime \rightarrow \{0,1\}$
({\em presence function}) indicates whether a given edge is {\em
  present} at a given time; and $\zeta: E \times \lifetime \to \T$
({\em latency function}), indicates the time it takes to cross a given
edge ({\it i.e.,} send a message) if starting at a given time.
 In this paper we assume $\zeta$ to be constant over all edges and dates, and call
$\zeta$ the {\em crossing delay};  thus we use the shorthand notations
$\G=(V,E,\lifetime,\rho)$. We also assume 
 that, for every edge $e\in E$, the union of dates when
$\rho(e) = 1$, is a set of disjoint closed time intervals of length at least $\zeta$.
 Finally, the (static) graph formed by $V$
and $E$, taken alone, is the {\em footprint} of $\G$ (also called {\em
  underlying graph} or {\em interaction graph}). In this work, we do
not consider the footprint to be a complete graph in general (some
nodes may never interact), but we consider it to be connected.

The TVG formalism essentially encompasses that of {\em evolving graphs}~\cite{Fer04}, where $\G$ is represented as a sequence of 
graphs $G_1, G_2, ..., G_i, ...$ each providing a {\em snapshot} of the system at different times (which correspond either to discrete steps or given dates).
In comparison, TVGs offer an {\em interaction-centric} view of the network evolution, where the evolution of each edge can be considered irrespective of the global time sequence.

A graph $G$ is said to be {\em recurrent} if none of the edges $E$ can
disappear forever; that is, for any date $t$ and edge $e$,
$\rho(e,t)=0 \implies \exists t'>t : \rho(e,t')=1$. Strictly speaking,
we do not say that an edge must appear infinitely often because here an edge might also remain present
continuously (and this would satisfy the property). Let $\cal{R}$
denote the class of recurrent TVGs whose footprint $G=(V,E)$ is {\em
  connected}.

A graph $\G\in \RG$ is said to be {\em time-bounded recurrent} (or
simply {\em bounded}), if there exists a constant $\Delta$ such that,
for every edge $e \in E$, the time between two successive appearances
of $e$ is at most $\Delta$. We denote by $\cal{B} \subset \cal{R}$ the
class of time-bounded recurrent TVGs whose footprint is connected.

A graph is said to be {\em periodic } if there exists a constant $p$ such that $\forall e\in E$, $\rho(e,t)=\rho(e,t+kp)$ for every positive  integer $k$; the smallest such $p$ is called the {\em period} of the graph.
 We denote by $\PG\subset\BG$ the class of periodic TVGs whose footprint is connected.

Given a TVG $\G=(V,E,\lifetime,\rho)$, we consider that $G=(V,E)$ is always simple (no self-loop nor multiple edges) 
and that
nodes possess unique identifiers.

The set of edges being incident to a node $u$ at time $t$ is noted $I_t(u)$ (or simply $I_t$, when the node is implicit). 

When an edge $e=(x,y)$ appears, the entities $x$ and $y$ can communicate. 
The time $\zeta$ necessary to transmit a message ({\em crossing delay}) is known to the nodes. 
The   duration of edge presence  is assumed to be  at least $\zeta$ ({\it i.e.,} long enough to send a message). 
Algorithmically, this allows  the following observations:

\begin{prop}~\\
1. If a message is sent just after an edge has appeared, the message is guaranteed to be successfully transmitted.\\
2. If the recurrence of an edge is bounded by some $\Delta$, then this edge cannot disappear for more than $\Delta-\zeta$.
\label{prop:sendappear}
\end{prop}

 The appearance and disappearance of edges are instantly detected by the two adjacent nodes (they are notified of such an event without delay). If a message is sent less than $\zeta$ before the disappearance of an edge, it is lost. However, since the disappearance of an edge is detected instantaneously, and the crossing delay $\zeta$ is known, the sending node can locally determine whether the message was successfully delivered. We thus authorize the special primitive $send\_retry$ as a facility to specify that if the message is lost, then it is automatically re-sent upon next appearance of the edge, and this sending is necessarily successful (Property~\ref{prop:sendappear}). Note that nothing precludes this primitive to be called while the corresponding edge is even absent (this actually simplifies the expression of some algorithms).

A sequence of couples $\J=\{(e_1,t_1), (e_2, t_2),...,(e_k, t_k)\}$,
with $e_i \in E$ and $t_i \in \lifetime$ for all $i$, is called a {\em journey} in $\G$ iff $\{e_1, e_2,...\}$ is a walk in $G$ and  for all $t_i$,
   $t_{i+1} \geq   t_i+\zeta$ and 
$\rho(e_i)_{t'}=1$ for all $  t_i  \leq  t'  \leq  t_i+\zeta$.
 We denote by $departure(\J)$, and $arrival(\J)$, the starting date $t_1$ and last date $t_k +\zeta$ of $\J$, respectively.

Journeys can be thought of as {\em paths over time} from a source node to a destination node (if the journey is finite). Let us denote by $\J^*_\G$ the set of all finite journeys in a graph $\G$. We will say that node $u$ can reach node $v$ in $\G$, and note $\exists \J_{(u,v)} \in \J^*_\G$ (or simply $u \leadsto v$, if $\G$ is clear from the context), if there exists at least one possible journey from $u$ to $v$ in $\G$. Note that the notion of journey is asymmetrical ($u \leadsto v \nLeftrightarrow v \leadsto u$), regardless of whether edges are directed or undirected.

Because journeys take place {\em over time}, they have both a topological length and a temporal length. 
The {\em topological length} of $\J$ is the number $|\J|_h=k$ of couples in $\J$
(i.e., number of {\em hops}), and its {\em temporal length} is its duration $|\J|_t  =  arrival(\J) - departure(\J) = t_k - t_1 +\zeta$. 
This yields two distinct definitions of distance in a graph~$\G$:

\begin{list}{\labelitemi}{\leftmargin=.5em}
  \item The {\em topological distance} from a node $u$ to a node $v$ at time $t$, noted 
  $d_{u,t}(v)$, is defined as $Min\{|\J|_h:\J \in \J^*_{(u,v)} \wedge departure(\J) \ge t\}$. 
  For a given date $t$, a journey whose departure is $t'\ge t$ and topological length is equal to $d_{u,t}(v)$ is called {\em shortest};
  
  \item The {\em temporal distance} from $u$ to $v$ at time $t$, denoted by $\hat{d}_{u,t}(v)$ is defined as $Min\{arrival(\J):\J \in \J^*_{(u,v)} \wedge departure(\J)\ge t\}-t$. Given a date $t$, a journey whose departure is $t'\ge t$ and arrival is $t+{\hat d}_{u,t}(v)$ is called {\em foremost};  
 if the set  ${ \hat d}_{u,t'}(v) $ has a minimum,  say $d'$,  any journey whose temporal distance is $d'$ is called the {\em fastest}.
\end{list}

Informally, a {\em foremost} journey is one that minimizes the date of arrival at destination; a {\em shortest} journey is one that uses the least number of hops; and a {\em fastest} journey is one that minimizes the time spent between departure and arrival (however late the departure be)~\cite{BFJ03}.


\subsection{Problems}


We consider the {\em distributed}  problem of  performing {\em broadcast with termination detection} at the emitter, 
or \tdbroadcast, according to the shortest, fastest, or foremost metrics.


\tdbroadcast in general requires all nodes to receive a message 
with some information initially held by a single node $x$, 
called {\em source} or {\em emitter}, 
and the source to enter a terminal state
after all nodes have received the information, within finite time.
 A protocol solves \tdbroadcast in a graph $\G$ if  it solves it for any source $x\in V$ and time $t\in \lifetime$.
 We say that it solves \tdbroadcast in a class $\cal{C}$ if it solves  \tdbroadcast  for any $\G \in \cal{C}$. We are interested in three variations of this problem, following the optimality metrics defined above: 
\begin{itemize}
\item \foremost, where {\em each} node receives the information at the {\em earliest} possible date following its {\em creation} at the emitter;
\item \shortest, where each node receives the information within a minimal number of hops from the emitter;
\item \fastest, where the overall duration between first global emission and last global reception is minimized. 
\end{itemize}
For each of these problems, we require that the emitter detects termination, however this detection is not subject to the same optimality constraint (it just has to be finite).
\tdbroadcast thus involves two processes: the actual 
dissemination of {\em information messages}, and the exchange of typically smaller {\em control messages} used for termination detection, both being considered separately in this paper.

Finally, we call {\em broadcast tree} the hierarchy of nodes and edges
along which the broadcast takes place, without consideration to the
dates when the edges are used ({\it i.e.} the footprint of the part of
the TVG that is used). A broadcast tree is said to be {\em reusable}
if the same hierarchy of nodes and edges can be purposedly followed to
perform a subsequent optimal ({\it i.e.} foremost, shortest, or
fastest) broadcast.

\section{Basic Results and Limitations}
\label{sec:basic}

Let us first state a general property on the computational
relationship between the main three contexts of interest, namely
knowing $n$ in \RG (noted $\RG_n$), knowing $\Delta$ in \BG (noted
$\BG_\Delta$), and knowing $p$ in \PG (noted $\PG_p$). Let $\PP({\cal
  C}_K)$ denote the set of problems solvable in every ${\cal
  G}\in{\cal C}$ with knowledge $K$.

\begin{theorem}
  $\PP(\RG_n)\subseteq\PP(\BG_\Delta)\subseteq\PP(\PG_p)$
  \label{th:inclusion}
\end{theorem}
\begin{proof}
  The right inclusion is straight from the fact that $\BG \subset \PG$ and $p$ is a valid bound $\Delta$ on the recurrence time.
The left inclusion follows from the facts that $\RG \subset \BG$ and $n$ can be inferred in \BG if $\Delta$ is already known. This can be done by performing, from any node (say $u$), a depth-first token circulation that will explore the underlying graph $G$ over time. Having a bounded recurrence time indeed allows every node to learn the list of its neighbors in $G$ within $\Delta$ time (all incident edges must appear within this duration). As the token is circulated to unvisited nodes, these nodes are marked as visited by $u$'s token and the token is incremented. Upon returning to $u$, the token value is $n$.
\end{proof}

\noindent These inclusions will be shown strict later on.

We now establish a negative result that justifies the need for additional knowledge in order to solve \tdbroadcast in any of the considered contexts. In fact we have:
 
\begin{theorem}
\tdbroadcast cannot be solved in  $\cal{P}$ without additional knowledge.
 \label{th:imp1}
\end{theorem}
\begin{proof}
 By contradiction, let ${\cal A}$ be an algorithm that solves  \tdbroadcast in \PG .
 Consider an arbitrary  $\G=(V,E,\lifetime,\rho)\in \PG$ and $x\in V$. Execute
  ${\cal A}$ in $\G$ starting at time $t_0$ with $x$ as the source.  Let $t_f$  be the
time when
 the source terminates (and thus all nodes have received the information).
 Let
$\G'=(V',E',\lifetime',\rho')\in \PG$ such that
 $V'=V\cup \{v\}$,
 $E'=E\cup\{(u,v)$ for some $u \in V\}$,
for all $ t_0\leq t < t_f$, 
 $\rho' (e,t)=\rho(e,t)$  for all $e\in E$ and $\rho'((u,v),t) = 0$.
 Now, consider $\rho'((u,v),t) = 1$ for some $t > t_f$, and the period of $\G'$ is some $p'>t-t_0$. 
Consider the execution of ${\cal A}$ in $\G'$  starting at time
$t_0$ with $x$ as the source.
 Since  $(u,v)$ does not appear from $t_0$ to $t_f$, the execution of
$ {\cal A}$ at every node in $\G'$ is exactly as at the corresponding node in $\G$.
In particular, node $x$ will have entered a terminal state at time  $t_f$ with node $v$
not having received the information, contradicting the correctness  of ${\cal A}$. 
  \end{proof}
  
 We thus have the following corollary, by inclusion of \PG. 

\begin{corollary}
 \tdbroadcast cannot be solved in \BG nor \RG without any additional knowledge. 
 \label{co:imp2}
\end{corollary}

Hence, additional knowledge of some kind is required to solve
\tdbroadcast in these classes. We consider three types of knowledge,
namely, the number of nodes $n=|V|$, an upper bound $\Delta$ on the
recurrence time (when in $\cal{B}$), or the period $p$ (in \PG).

%

To prove some impossibility results 
on these problems with a given knowledge (Theorems  \ref{th:no-reuse-P-ndelta} and \ref{th:neg-shortest-P}, later in the paper),
we  make use of a specific family of TVGs with the same footprint, and establish some basic limitations they expose.

Consider the graph $G = (V,E)$ with $V=\{u,v,x,y\}$ and 
$E=\{e_1=(u,x), e_2= (u,y), e_3= (x,v), e_4= (y,v)\}$. 
 Consider   the   infinite family   $\{{\cal G}_i\}= \{{\cal G}_0, {\cal G}_1, \ldots  \}$     of periodic  TVGs
 with footprint $G$ where, for each ${\cal G}_i$,   $\zeta=1$ and $\rho$ is as follows, where $[t]_j$ denotes $t$ modulo $j$:
\begin{itemize}
\item    $e_1$ and $e_2$ are present only during the intervals $[t,t+1]$ with $ t  \in \N$ and  $ [t]_4 = 0$;
\item $e_3$  is present only during the   intervals $[t,t+1]$ where $t \in \N$ and  $[t]_4 = 2$;

\item    If $i=0$, 
then $e_4$  is present only during the   intervals $[t,t+1]$ where $t \in \N$ and  $[t]_4 = 3$.

\item If $i>0$, 
then $e_4$   is present only during the   intervals $[t,t+1]$  where $   t \in \N$ and   either    $[t]_4 = 3$
or   $[t]_{4(i+1)} = 4i+1$.

 \end{itemize}

Notice that   ${\cal G}_i\in \PG$ and its period is  $4 (i+1) $.
Also notice that
 in all of these graphs, $n=4$ and the minimum  $\Delta$ is   4.

Let  {\sc Test}($\{{\cal G}_i\}$) be the problem of 
observing the evolution of a graph  chosen by an adversary from
 $\{{\cal G}_i\}$, and deciding  in finite time whether  or not  
 it   is  ${\cal G}_0$. 
 
  \begin{theorem}
 {\sc Test}($\{{\cal G}_i\}$) is undecidable.
\label{undecidable}
 \end{theorem}
\begin{proof}
At any time    $t$,  the evolution of
 $  {\cal G}_0$ from time $0$ to  $t$  is indistinguishable from that of any 
${\cal G}_i$ with  $i >  {{ t -1}\over 4}$. Since any solution algorithm must terminate in finite time, any decision taken 
at that time, say $\hat t$,  can be made incorrect by the adversary by choosing the graph  to be
 any  ${\cal G}_i$ with $i >  {{ \hat t -1}\over 4}$ if the answer was  $  {\cal G}_0$, and  by choosing it  to be $ {\cal G}_0$  otherwise.

\end{proof}

 A  knowledge $K$ about   $\{ {\cal G}_i  \}$ may reduce the choices of the  adversary.
   Indeed, given enough knowledge, the  {\sc Test}   problem can  become decidable.  
   The following lemma identifies   conditions on $K$ for {\sc Test}  to be undecidable.
   
Let  
  $\{ {\cal G}_i  \} \setminus K$ denote the subset of  $\{ {\cal G}_i  \}$ 
  still available to the adversary in spite of  knowledge 
$K$.

\begin{theorem}
If   ${\cal G}_0 \in \{ {\cal G}_i  \} \setminus K$ and  $|\{ {\cal G}_i  \} \setminus K| =\infty$, then
   {\sc Test}$(\{      {\cal G}_i  \}\setminus K) $  is undecidable.
\label{lemma}
\end{theorem}

 \begin{proof}
 Since  $|\{ {\cal G}_i  \} \setminus K| =\infty$,  
for any $t$  there is always a TVG with $i> {{ t -1}\over 4}$ in the set  $  \{ {\cal G}_i\} \setminus K$,
and its evolution from time $0$ to  $t$  is indistinguishable from that of  
 $  {\cal G}_0$. Since any solution algorithm must terminate in finite time, 
  any decision taken 
at that time, say $\hat t$, can be made incorrect by the adversary
by
choosing the graph  to be  any  ${\cal G}_j \in \{ {\cal G}_i\}  \setminus K $ with $j >  {{\hat  t -1}\over 4} $ if the answer was  $  {\cal G}_0$, and  by choosing it  
 to be $ {\cal G}_0$  otherwise.\\

\end{proof}

As an immediate consequence, we have that  $n$ and $\Delta$ are not sufficient:

\begin{lemma}{\sc Test}$(\{ {\cal G}_i  \}\setminus \{n,\Delta\})$ is undecidable.
\label{newlemma}
\end{lemma}

\begin{proof}
Let $n$ and $\Delta$ be known.
Since $n$ and $\Delta$ are the same for every ${\cal G}_i$, then
$\{ {\cal G}_i\} \setminus n = \{ {\cal G}_i\} \setminus \Delta  =  \{ {\cal G}_i\} \setminus \{n,\Delta\} = \{ {\cal G}_i\}$. 
Thus, by Theorem \ref{undecidable},  the lemma follows.
\end{proof}

%
%
%
%

\section{\foremost}

Solving \foremost in \RG or \BG clearly requires some sort of flooding, because the very fact of probing a neighbor to determine if it already has the information compromises the possibility to send it in a foremost fashion (in addition to risking the disappearance of the edge in-between the probe and the real sending). 
As a consequence of Theorem~\ref{th:imp1}, this problem cannot be solved without knowledge.
In this section we first show that it becomes possible in \RG if the 
number of nodes $n=|V|$ is known. The proof is constructive by means of Algorithm~\ref{algo:foremost-recurrent-n}, whose termination is however not bounded in time. Being in \BG with the same knowledge allows its termination to be bounded. Knowing $\Delta$ instead of $n$ in \BG then allows us to propose another solution  (described in Algorithm~\ref{algo:foremost-recurrent-delta}) that has a lower message complexity. This complexity can be further improved if both $\Delta$ {\em and} $n$ are known, as in this case we have the possibility to terminate {\em implicitly}. 

%
%
%

\subsection{\foremost in  $\cal{R}$}
\label{sec:Rn}
Since $\Delta$ and $p$ are not defined for  \RG, we need to focus only on   the knowledge of $n$.
We  show that the problem is solvable when $n$ is known. 

The algorithm proceeds as follows (see Algorithm~\ref{algo:foremost-recurrent-n} for details). Every time a {\em new} edge appears locally to an informed node, this node sends the information message onto this edge, and remembers that this edge now leads to an informed node.
The first time a node receives the information, it records the sender as parent, transmits the information on its available edges, and sends back a notification message to the parent. 
Note that these notifications create a parent-relation and thus a converge-cast tree.
Each notification is propagated along the converge-cast tree and eventually collected at the emitter. When the emitter has received $n-1$ notifications, it knows all nodes are informed.
Observe that the notification messages are sent using the special primitive $send\_retry$ discussed in Section~\ref{sec:defs}, to ensure that the parent eventually receives it even if the edge disappears during the first attempt. Information messages, on the other hand, are sent using the normal $send$ primitive. Indeed, if the propagation of such a message fails because the corresponding edge disappears, it simply means that this edge at that particular time did not have to be used (i.e., it did not belong to a valid journey).

\begin{algorithm*}[h]
  \small
  \begin{spacing}{1}
  \begin{code}[1]\medskip
    \STATE \hspace{-8pt}$Edge\ parent \gets nil$
    \COMMENT {edge the information was received from (for non-emitter nodes).}
    \STATE \hspace{-8pt}$Integer\ nbNotifications \gets 0$
    \COMMENT {number of notifications received (for the emitter).}
    \STATE \hspace{-8pt}$Set$$<$$Edge$$> informedNeighbors \gets \emptyset$
    \COMMENT {neighbors \underline{known} to have the information.}
    \STATE \hspace{-8pt}$Status\ myStatus \gets \neg{\tt informed}$
    \COMMENT {status of the node (informed or non-informed).}
    \medskip
    \STATE \hspace{-8pt}\underline{\textbf{\textit{initialization}}:}\medskip
    \IF {$isEmitter()$}
    \STATE $myStatus \gets {\tt informed}$
    \STATE $send(information)\ on\ I_{now()}$
    \COMMENT {sends the information on all present edges.}
    \ENDIF
    \STATE \hspace{-8pt}\underline{\textbf{\textit{onAppearance}} of an edge $e$:}\medskip
    \IF {$myStatus=={\tt informed}$ {\bf and} $e \notin informedNeighbors$}
    \STATE $send(information)\ on\ e$
    \STATE $informedNeighbors \gets informedNeighbors \cup \{e\}$
    \COMMENT {(see Prop.~\ref{prop:sendappear}).}
    \ENDIF\medskip
    \STATE \hspace{-8pt}\underline{\textbf{\textit{onReception}} of a message $msg$ from an edge $e$:}\medskip
    \IF {$msg.type==Information$}
    \STATE $informedNeighbors \gets informedNeighbors \cup \{e\}$
    \IF {$myStatus==\neg{\tt informed}$}
    \STATE $myStatus \gets {\tt informed}$
    \STATE $parent \gets e$
    \STATE $send(information)\ on\ I_{now()}\smallsetminus informedNeighbors$
    \COMMENT {propagates.}
    \STATE $send\_retry(notification)\ on\ e$
    \COMMENT {notifies that this node has the info. \\\hfill (this message is to be resent upon the next appearance, in case of failure).}
    \ENDIF
    \ELSIF {$msg.type==Notification$}
    \IF {$isEmitter()$}
    \STATE $nbNotifications \gets nbNotifications + 1$
    \IF {$nbNotifications==n-1$}
    \STATE $terminate$
    \COMMENT {at this stage, the emitter knows that all nodes are informed.}
    \ENDIF
    \ELSE
    \STATE $send\_retry(notification)\ to\ parent$
    \ENDIF
    \ENDIF
  \end{code}
  \end{spacing}
  \caption{\label{algo:foremost-recurrent-n} Foremost broadcast in $\cal{R}$, knowing $n$.}
\end{algorithm*}

\begin{theorem}
  When $n$ is known, \foremost can be solved in $\cal{R}$ 
   exchanging $O(m)$ information messages and $O(n^2)$ control messages, in unbounded time. 
\label{theo:foremost-R-n-known}
\end{theorem}
\begin{proof}
Since a node sends the information to each  new appearing edge,
 it is easy to see, by   connectivity of  the {\em underlying} graph,  
 that all nodes will eventually receive the information. 
 The dissemination itself is necessarily foremost because
 the information is either directly relayed on edges that are present, or sent as soon as a new  edge  appears.
As for termination detection:
 every node identifies a unique parent and a converge-cast spanning tree 
 directed towards the source is implicitly  constructed; 
 since every node notifies the source (through the tree)
 and the source knows the total number of nodes, termination is guaranteed.
Since information  messages might traverse every edge in both directions, and
an edge cannot be traversed twice in the same direction, we have that 
  the number of {\it information} messages   
is in the worst case $2m$. 
Since every node but the emitter induces a notification that is forwarded
up the converge-cast  tree to the emitter, the number of {\em notification} messages 
is  the sum of distances in the converge-cast  tree between all nodes and the emitter, 
$\sum_{v \in V\smallsetminus\{emitter\}} d_{h\_tree}(v, emitter)$. The worst case
is when the graph is a line where we have    $\frac{n^2-n}{2}$ control messages. 
Regarding time complexity, the termination of the algorithm is unbounded due to the fact that the recurrence of the edges is itself unbounded.
\end{proof}

\paragraph{Use in subsequent broadcasts.}

Foremost trees are {\em time-dependent} in the sense that they might
be optimal for some emission dates and not be so for other dates.
Still, they remain {\em valid} trees (though, possibly non-foremost
ones) regardless of the considered date. As such, they can be
memorized by the nodes in order to be used as {\em converge-cast}
trees for termination detection in subsequent broadcasts. Indeed,
while the broadcast is required to be foremost, the detection of
termination does not have such constraint. Hence, instead of sending a
notification each time a new node is informed (as done previously),
nodes can notify their parents (in the converge-cast tree) if and only
if they are themselves informed and have received a notification from
each of their children (in the converge-cast tree).
This reduces the
number of control messages from $O(n^2)$ to $O(n)$, having only one
notification per edge of the converge-cast tree.

\subsection{\foremost in  $\cal{B}$}
If the recurrence time is bounded, then either the knowledge of $n$ or an upper bound $\Delta$ on the recurrence time can be used to solve the problem (with various message complexities).

\subsubsection{Knowledge of $n$.}
Since $\BG \subseteq \RG$, one can obviously solve \foremost in $\cal{B}$ using Algorithm~\ref{algo:foremost-recurrent-n} (and the same observations apply regarding reusability of the converge-cast tree). Here, however, the termination time becomes bounded due to the fact that the recurrence of edges is itself bounded.

\begin{theorem}
  When $n$ is known, \foremost can be solved in $\cal{B}$ exchanging $O(m)$ information messages and $O(n^2)$ control messages, in $O(n\Delta)$ time.
\end{theorem}

\begin{proof}
  Since all edges in $E$ are recurrent within any $\Delta$ time window, the delivery of the information at the last node must occur within $(n-1)\Delta$ global time. The same property holds for the latest notification, bounding the overall process to a duration of $\Delta(2n-2)$. The rest follows from Theorem~\ref{theo:foremost-R-n-known}.
\end{proof}

\subsubsection{Knowledge of  $\Delta$.}
\label{sec:foremost-recurrent-delta}

The information dissemination is performed as in Algorithm
\ref{algo:foremost-recurrent-n}, but the termination detection is
different. Thanks to the time-bound $\Delta$ on edge recurrence, a
node can discover all of its neighbors within $2\Delta$ time (back and
forth messages, each on a different $\Delta$ period at the worst).
This fact can be used by a node to determine whether it is a leaf in
the broadcast tree ({\it i.e.,} if it has not heard back from a
potential child within $2\Delta$ time following its own reception
time). This allows the leaves to terminate spontaneously and notify
their parent, which recursively terminate after receiving the
notifications from all their children and notifying their own parent.
This combination of broadcast-convergecast, originally described
in~\cite{DS80}, is a standard process for distributed computing in
{\em static} graphs. The temporal adaption for bounded-recurrent TVGs
is as follows (see Algorithm~\ref{algo:foremost-recurrent-delta} for
details). First, everytime a {\em new} edge appears locally to an
informed node, this node sends the information on the edge, and
records the edge. The first time a node receives the information, it
chooses the sender as parent, memorizes the current time (in a
variable $firstRD$), transmits the information on its available edges,
and returns an {\em affiliation} message to its parent using the
$send\_retry$ primitive (starting to build the converge-cast tree).
This affiliation message is not relayed further up the tree: it is
only intended to inform the parent about the existence of a new child
(so this parent knows it must wait for a future notification by this
node). If an informed node has not received any affiliation message
after a duration of $2\Delta$, then it sends a {\em notification}
message to its own parent using the $send\_retry$ primitive.

As for Algorithm~\ref{algo:foremost-recurrent-delta}, if the information message is lost, then it simply means that this edge at that time did not have to be used. On the other hand, if the {\em affiliation} message is lost, it must be sent again ($send\_retry$). However, in the worst case, the common edge disappears just before the affiliation message is delivered, and reappears less than $\Delta - \zeta$ later (Prop.~\ref{prop:sendappear}). The overall back and forth exchange thus remains within $2\Delta$ time. 

If a node has one or more children, it waits until it receives a notification message from each of them, then notifies its parent in the converge-cast tree (using $send\_retry$ again). Once the emitter has received a notification from each of its children, it knows that all nodes are informed.

\begin{algorithm*}[h]
  \small
  \begin{spacing}{1}
  \begin{code}[1]\medskip
    \STATE \hspace{-8pt}$Edge\ parent \gets nil$
    \COMMENT {edge the information was received from (for non-emitter nodes).}
    \STATE \hspace{-8pt}$Integer\ nbChildren \gets 0$
    \COMMENT {number of children.}
    \STATE \hspace{-8pt}$Integer\ nbNotifications \gets 0$
    \COMMENT {number of children that have terminated.}
    \STATE \hspace{-8pt}$Set$$<$$Edge$$> informedNeighbors \gets \emptyset$
    \COMMENT {neighbors \underline{known} to have the information.}
    \STATE \hspace{-8pt}$Date\ firstRD \gets nil$
    \COMMENT {date of first reception.}
    \STATE \hspace{-8pt}$Status\ myStatus \gets \neg{\tt informed}$\medskip
    \COMMENT {status of the node (informed or non-informed).}
    \STATE \hspace{-8pt}\underline{\textbf{\textit{initialization}}:}\medskip
    \IF {$isEmitter()$}
    \STATE $myStatus \gets {\tt informed}$
    \STATE {$send(information)\ on\ I_{now()}$}
    \COMMENT {sends the information on all present edges.}
    \ENDIF
    \STATE \hspace{-8pt}\underline{\textbf{\textit{onAppearance}} of an edge $e$:}\medskip
    \IF {$myStatus=={\tt informed}$ {\bf and} $e \notin informedNeighbors$}
    \STATE $send(information)\ on\ e$
    \STATE $informedNeighbors \gets informedNeighbors \cup \{e\}$
    \COMMENT {(see Prop.~\ref{prop:sendappear}).}
    \ENDIF\medskip
    \STATE \hspace{-8pt}\underline{\textbf{\textit{onReception}} of a message $msg$ from an edge $e$:}\medskip
    \IF {$msg.type==Information$}
    \STATE $informedNeighbors \gets informedNeighbors \cup \{e\}$
    \IF {$myStatus==\neg{\tt informed}$}
    \STATE $myStatus \gets {\tt informed}$
    \STATE $firstRD \gets now()$
    \COMMENT {memorizes the date of first reception.}
    \STATE $parent \gets e$
    \STATE $send(information)\ on\ I_{now()}\smallsetminus informedNeighbors$
    \COMMENT {propagates.}
    \STATE $send\_retry(affiliation)\ on\ e$
    \COMMENT {informs the parent that it has a new child.}
    \ENDIF
    \ELSIF {$msg.type == Affiliation$}
    \STATE $nbChildren \gets nbChildren + 1$
    \STATE $informedNeighbors \gets informedNeighbors \cup \{e\}$
    \ELSIF {$msg.type == Notification$}
    \STATE $nbNotifications \gets nbNotifications + 1$
    \IF {$nbNotifications == nbChildren$}
    \IF {$\neg isEmitter()$}
    \STATE $send\_retry(notification)\ to\ parent$
    \COMMENT {notifies the parent in turn.}
    \ENDIF
    \STATE $terminate$
    \COMMENT {whether emitter or not, the node has terminated at this stage.}
    \ENDIF
    \ENDIF\medskip
    \STATE \hspace{-8pt}\underline{\textbf{\textit{when}} $now() == firstRD + 2\Delta$:}\medskip
    \COMMENT {tests whether the underlying node is a leaf.}
    \IF {$nbChildren == 0$}
    \STATE $send\_retry(notification)\ on\ parent$
    \STATE $terminate$
    \ENDIF
  \end{code}
  \end{spacing}
  \caption{\label{algo:foremost-recurrent-delta} Foremost broadcast in $\cal{B}$, knowing a bound $\Delta$ on the recurrence time.}
\end{algorithm*}

\begin{theorem}
  When $\Delta$ is known, \foremost can be solved in $\cal{B}$ exchanging $O(m)$ information messages and $O(n)$ control message, in $O(n\Delta)$ time.
\end{theorem}
\begin{proof}
Correctness follows the same lines of the proof of Theorem \ref {theo:foremost-R-n-known}.
However the correct construction of a converge-cast spanning tree is guaranteed by
 the knowledge of $\Delta$ (i.e., the nodes of the tree that are leaves detect their status because 
 no new edges appear within $\Delta$ time)  and the notification starts from the 
 leaves and is aggregated before reaching the source.
  The number of information messages is $O(m)$ as the exchange of information messages is the same as in Algorithm \ref{algo:foremost-recurrent-n}, but the number of notification and affiliation messages decreases to
  $2(n-1)$. Each node but the emitter sends a single affiliation message;
  as for the notification messages,  instead of sending a notification as soon as it is informed,  
each node notifies its parent in the converge-cast tree 
if and only if it  has received a notification from each of its children
resulting in one notification message per edge of the tree. 
  The time complexity of the dissemination itself is the same as for the foremost broadcast when $n$ is known ($(n-1)\Delta$). The time required for the emitter to subsequently detect termination is an additional $2\Delta + (n-1)\Delta$ in the worst case ({\it i.e.} self-detection by the leaves, followed by the longest notification chain up the emitter). This gives a total of $2n\Delta$.
\end{proof}


 Clearly, the number of nodes $n$, which is not {\it a priori} known here, can be obtained through the notification process of the first broadcast (by having nodes reporting their number of descendants in the tree, while notifying hierarchically).
All subsequent broadcasts can thus behave as if both $n$ and $\Delta$ were known.
Next we show this allows to solve the problem without any control messages.
 
\subsubsection{Knowledge of both  $n$ and $\Delta$}\label{both}

In this case, the emitter knows an upper bound on the broadcast termination date;
in fact, the broadcast cannot last longer than $ (n-1) \Delta$ (this worst case is when 
the foremost tree is a line). 
Termination detection can thus become implicit after this amount of time, 
which removes the need for any control message (whether of affiliation or of notification). 

\begin{theorem}
 When $\Delta$ and $n$  are known, \foremost 
 can be solved in $\cal{B}$ exchanging $O(m)$ information messages and no 
control messages,  in $O(n \Delta)$ time.
\end{theorem}

\subsection{Reusability}

As we have seen, \foremost is feasible even in $\RG$ with knowledge $n$. 
However our algorithms in $\RG_n$ or even  in $\BG_{\{n,\Delta\}}$  do not provide  reusability.
This is not accidental; in fact, 
 we will  now show that achieving reusability in \foremost
 is {\em  impossible} even in   $\PG_{\{n,\Delta\}}$.

\begin{theorem}  
Foremost broadcast trees are not reusable in $\PG_{\{n,\Delta\}}$. 
\label{th:no-reuse-P-ndelta}
\end{theorem}
\begin{proof}
Consider the  infinite family   $\{{\cal G}_i\}$ of TVGs defined in Section \ref{sec:basic}.
 In  ${\cal G}_0$ the foremost journey from $u$ to $v$ at time $t$    is always along the edges $e_1$ and $e_3$.
 In  ${\cal G}_i$, $i>0$,  the foremost journey  from $u$ to $v$ at time $t>0$ with   $[t]_{4i}=0$
 is  also along the edges $e_1$ and $e_3$, 
{\em except} when  $[t]_{4(i+1)}=0$, in which case  the foremost journey  is  along the
edges $e_2$ and $e_4$.
By Lemma \ref{newlemma},  even knowing $n$ and $\Delta$, the evolution of ${\cal G}_0$ until time $t$ 
with $[t]_{4i}=0$ is undistinguishable from that of  ${\cal G}_i$  with $i  \geq {t \over 4} $.
Then, by  just observing the evolution until time $t$ with $[t]_{4i}=0$, it is impossible to decide at that time
 whether the graph is   ${\cal G}_0$ or not;  hence, it is impossible to decide 
 whether  the foremost path from $u$ to $v$ at that time  is  $\{e_1,e_3\}$ (foremost in    ${\cal G}_0$)
or   $\{e_2,e_4\}$ (foremost in    ${\cal G}_i$  with $i  \geq {t \over 4} $).
\end{proof}

While $n$ and $\Delta$ are not sufficient for reusability in \PG,  it has been shown in~\cite{CFMS14} that the 
knowledge of the period $p$ is sufficient.

\begin{theorem} \cite{CFMS14}
Foremost broadcast trees can be reused for subsequent broadcasts in $\PG$ with knowledge $p$. 
\label{th:reuse-P-p}
\end{theorem}

\noindent The basic argument is that a foremost broadcast tree for time $t$ remains optimal for all times $t+jp$, where $j$ is a positive integer.

\section{\shortest}

Recall that the objective of \shortest is to deliver the information to each node within a minimal {\em number of hops} from the emitter, and to have the emitter detect termination within finite time. 
We show below that contrary to the foremost case, knowing $n$ is insufficient to perform a shortest broadcast in $\RG$ or even in $\cal{B}$. This becomes however feasible in \BG when $\Delta$ is also known.
Moreover any shortest tree built at some time $t$ will remain optimal in $\BG$ relative to any future emission date $t'>t$. This feature allows the solution to \shortest to be possibly reused in subsequent broadcasts.

\subsection{\shortest in  $\cal{B}$}


As we will show later (see Theorem~\ref{th:neg-shortest-P}),  knowing $n$ is not sufficient to solve
 \shortest in \PG, and thus also in \BG  (or \RG). Therefore, the only knowledge still of interest are $\Delta$ and the combination of $\Delta$ and $n$.


\subsubsection{Knowledge of $\Delta$}
\label{sec:shortest-delta}

The idea is to propagate the message along the edges of a {\em breadth}-first 
 spanning tree
of the underlying graph. We present the pseudo-code in Algorithm \ref{algo:shortest-delta}, and provide the following informal description. 

Assuming that the message is created at some date $t$, the mechanism consists of authorizing nodes at level $i$ in the tree to inform new nodes only between time $t+i\Delta$ and $t + (i+1)\Delta$ (doing it sooner would lead to a non-shortest tree, while doing it later is pointless because all the edges have necessarily appeared within one $\Delta$). So the broadcast is confined into rounds of duration $\Delta$ as follows: whenever a node sends the information to another, it sends a time value that indicates the 
remaining duration of its round (that is, the starting date of its own round minus the current time, plus $\Delta$, minus the crossing delay $\zeta$), so the receiving node, if it is a new child, knows when it should start informing new nodes in its turn. 
For instance in Figure~\ref{fig:roundstart} when the node $a$ attempts to become $b$'s parent,  node a transmits its own starting date plus $\Delta$ minus the current date minus $\zeta$. This duration corresponds to the exact amount of time the child would have to wait, if the relation is established, before integrating other nodes in turn.
If a node does not detect any children before $2\Delta$ following its own reception date (same as Algorithm~\ref{algo:foremost-recurrent-delta}), then it detects that it is a leaf and notifies its parent. Otherwise, it waits for the final notifications of all its children, then notifies its parent. As before, this requires parents to keep track of the number of children they have, and thus children need to send {\em affiliation} messages when they select a parent. Finally, when the emitter has been notified by all its children, it knows the broadcast is terminated.

\begin{theorem}
  \shortest can be solved in $\cal{B}$ knowing $\Delta$, exchanging $O(m)$ info. messages and $O(n)$ control messages, in $O(n\Delta)$ time.
  \label{th:shortest-B}
\end{theorem}
\begin{proof} 
The fact that the algorithm constructs a breadth-first 
(and thus shortest) delay-tolerant spanning tree
follows from the connectivity over time of the underlying graph and from the knowledge of the duration $\Delta$. 
The bound  on recurrence is used to enable a rounded process whereby the correct distance of each node to the emitter is detected.
The number of {\em information} messages is $2m$ as the dissemination process 
exchanges at most two messages per edge. 
The number of {\em affiliation} and {\em notification} messages are each of $n-1$ 
(one per edge of the tree). 
The time complexity for the construction of the tree is at most $(n-1)\Delta$ to reach the last node, plus $2\Delta$ at this node, plus at most $(n-1)\Delta$ to aggregate this node's notification. (The additional $\zeta$ caused by waiting affiliation messages matters only for the last round, since the construction continues in parallel otherwise.) The total is thus at most $2n\Delta$. 
\end{proof}

\begin{algorithm*}[h]
  \begin{spacing}{1}
  \begin{code}[1]
    \medskip
    \STATE \hspace{-12pt}$Edge\ parent \gets nil$
    \COMMENT {edge the information was received from (for non-emitter nodes).}
    \STATE \hspace{-12pt}$Date\ roundStart \gets +\infty$
    \COMMENT {date when this node starts informing new nodes.}
    \STATE \hspace{-12pt}$Set<$$Edge$$> children \gets \emptyset$
    \COMMENT {set of children from which a notification is expected.}
    \STATE \hspace{-12pt}$Integer\ nbNotifications \gets 0$
    \COMMENT {number of children that have sent their notification.}
    \STATE \hspace{-12pt}$Set$$<$$Edge$$> informedNeighbors \gets \emptyset$
    \COMMENT {set of neighbors \underline{known} to have the info.}
    \STATE \hspace{-12pt}$Status\ myStatus \gets \neg {\tt informed}$
    \COMMENT {status of the node (informed or non-informed).}

    \codetitle{\textbf{\textit{initialization}}}
    \IF {$isEmitter()$}
    \STATE $roundStart \gets now()$
    \ENDIF

    \codetitle{\textbf{\textit{onAppearance}} of an edge $e$}
    \IF {$myStatus == {\tt informed}$}
    \IF {$e \notin informedNeighbors$}
    \STATE $send(roundStart - now() + \Delta - \zeta)\ on\ e$
    \COMMENT {time until end of round.}
    \STATE $informedNeighbors \gets informedNeighbors \cup \{e\}$
    \COMMENT {(see Prop.~\ref{prop:sendappear}).}
    \ENDIF
    \ENDIF
    
    \codetitle{\textbf{\textit{onReception}} of a message $msg$ from an edge $e$}
    \IF {$msg.type==Duration$}
    \STATE $informedNeighbors \gets informedNeighbors \cup \{e\}$
    \IF {$parent == nil$}
    \STATE $parent \gets e$
    \STATE $roundStart \gets now()+msg$
    \STATE $send\_retry(affiliation)\ on\ e$
    \ENDIF
    \ELSIF {$msg.type==Affiliation$}
    \STATE $children \gets children \cup \{e\}$
    \ELSIF {$msg.type==Notification$}
    \STATE $nbNotifications \gets nbNotifications + 1$
    \IF {$nbNotifications == |children|$}
    \IF {$\neg isEmitter()$}
    \STATE $send\_retry(notification)\ on\ parent$
    \ENDIF
    \STATE $terminate$
    \ENDIF
    \ENDIF

    \codetitle{\textbf{\textit{when}} $now()==roundStart$}
    \STATE $myStatus \gets {\tt informed}$
    \STATE $send(\Delta$-$\zeta)\ on\ I_{now()}\smallsetminus informedNeighbors$
    \COMMENT {nodes that receive this and have\\ \hfill  no parent yet will take this node as parent and wait $\Delta$-$\zeta$ before informing new nodes.}

    \codetitle{\textbf{\textit{when}} $now()==roundStart+2\Delta$}
    \COMMENT {tests whether the underlying node is a leaf.}
    \IF {$|children|==0$}
    \STATE $send\_retry(notification)\ on\ parent$
    \ENDIF
  \end{code}
  \end{spacing}
  \caption{\label{algo:shortest-delta}Shortest broadcast in $\cal{B}$, knowing a bound $\Delta$ on the recurrence.}
\end{algorithm*}

\newcommand{\vtick}[1]{
  \path (#1)+(0,.15) coordinate (tmp);
  \draw (tmp)+(0,-.3)--(tmp);
}
\begin{figure*}[h]
\centering
\begin{tikzpicture}[yscale=1]
  \path (-1,1) coordinate (a);
  \path (a)+(11,0) coordinate (aend);
  \path (-1,0) coordinate (b);
  \path (b)+(11,0) coordinate (bend);
  \path (a) node[left] {$a$};
  \path (b) node[left] {$b$};
  \path (0,1) node[above, yshift=6pt] {$roundStart$};
  \path (3,1) node[above, yshift=4pt] {$now()$};
  \path (8,1) node[above, yshift=4pt] {$roundStart+\Delta$};
  \path (8,0) node[above, yshift=5pt] {$roundStart$};
  \draw[<->, semithick] (3,-.2)-- node[below, midway] {$\zeta$} (4,-.2);
  \draw (a)--(aend);
  \draw (b)--(bend);
  \vtick{0,1};
  \vtick{3,1};
  \vtick{8,0};
  \vtick{8,1};
  \draw[dashed] (3,1)--(3,0);
  \draw[->, thick] (3,1)--(4,0);
  \draw[fill=gray, style=nearly transparent] (4,-.06) rectangle (8,.06);
\end{tikzpicture}
\caption{\label{fig:roundstart} Propagation of the rounds of duration $\Delta$. }
\end{figure*}
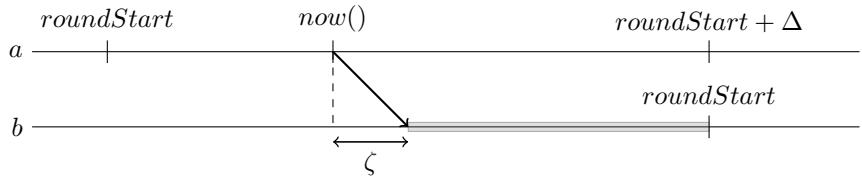

\FloatBarrier

\paragraph{Reusability for subsequent broadcasts:}
Thanks to the fact that shortest trees remain shortest regardless of the emission date, all subsequent broadcasts can be performed within the same, already known tree, which reduces the number of information message from $O(m)$ to $O(n)$. Moreover, if the depth $d$ of the tree is detected through the first notification process, then all subsequent broadcasts can enjoy an implicit termination detection that is itself optimal in time (after $d \Delta$ time). No control message 
is needed.

\subsubsection{Knowledge of  $n$ and $\Delta$}

When both $n$ and $\Delta$ are known,
one can apply the same dissemination procedure as in Algorithm \ref{algo:shortest-delta} combined with an implicit termination detection that avoids using control messages at all. Indeed, each node learns (and possibly informs) all of its neighbors within $\Delta$ time. Since the underlying graph is connected, the whole process must then complete within $n\Delta$ time. Hence, if the emitter knows both $\Delta$ and $n$, it can simply wait $n\Delta$ time, then terminate implicitly.

\begin{theorem}
  When $n$ and $\Delta$ are  known, \shortest can be solved in $\cal{B}$  exchanging $O(m)$ info. messages and  no control messages, in $O(n\Delta)$ time.
\end{theorem}

However, such a strategy would prevent the emitter from learning the depth $d$ of the shortest tree, and thus prevent lowering the termination bound to $d\Delta$ time. An alternative solution would be to achieve  explicit 
termination for the first broadcast in order to build a reusable broadcast tree 
(and learn its depth $d$ in the process). 
In  this case,  dissemination is achieved   with $O(m)$ 
information messages,
termination detection is achieved similarly to  Algorithm \ref{algo:shortest-delta}    with $O(n)$
control messages (where however affiliation messages are not necessary,
 and the number of control messages would decrease to $n-1$). 
 In this way  we would have an increase in control messages, but the  
 subsequent broadcasts  could reuse
the broadcast tree for dissemination with $O(n)$ information messages,
and termination detection could be implicit with no exchange of control message at all
after $d\Delta$ time. 
The choice of either solution may depend on the size of an information message and on the expected number of broadcasts planned.

\subsection{\shortest in  $\cal{P}$}

Feasibility in $\PG_\Delta$  is  implied by  feasibility in $\BG_\Delta$  and  it clearly implies feasibility in  $\PG_p$.
The only case left to study is   feasibility with only knowledge of $n$.


\begin{theorem}
  \shortest is not feasible in \PG knowing only $n$.
  \label{th:neg-shortest-P}
%
\end{theorem}
\begin{proof}
By contradiction, let ${\cal A}$ be an algorithm that solves  \shortest 
in  $\PG$ with the knowledge of $n$ only.
 Consider an arbitrary  $\G=(V,E,\lifetime,\rho)\in  \PG$ and $x\in V$. Execute
  ${\cal A}$ in $\G$  starting at time $t_0$ with $x$ as the source.  
  Let $t_f$  be the time when
 the source terminates and $T$ the shortest broadcast tree along which broadcast was performed.
 Let $\G'=(V',E',\lifetime',\rho')\in \PG$ such that
 $V'=V$,
 $E'=E\cup\{(x,v)$ for some $v \in V : (x,v) \notin E\}$,
 $\rho' (e,t)=\rho(e,t)$ for all $e\in E, 0 \le t \le t_f$,
 $\rho'((x,v),t) = 0$ for all
  $ t_0\leq t <  t_f$, and  $\rho'((x,v),t) = 1$ for some $ t >
t_f$ (we can take the period $p$ as large as needed here). 
Consider the execution of ${\cal A}$ in $\G'$  starting at time
$t_0$ with $x$ as the source.
 Since  $(x,v)$ does not appear between $t_0$ and $t_f$, the execution of
$ {\cal A}$ at every node in $\G'$ will be exactly as at the corresponding node in $\G$
and terminate with   $v$ having received the information 
in more than one hop, contradicting the fact that $T$ is a shortest tree, and thus the correctness of ${\cal A}$.
\end{proof}

\section{\fastest}

The requirement for  \fastest  is to deliver the information to each node using the least amount of time, regardless of the stating date, and to have the emitter detect termination within finite time. 

We first show  that, contrary to the foremost and shortest cases, knowing $n$ and $\Delta$  is insufficient to perform a fastest broadcast even in $\PG$.

\begin{theorem}
 \fastest is not feasible in \PG  with  only knowledge  of $n$  and $\Delta$.
\label{th:neg-fastest-n-delta}
\end{theorem}
\begin{proof}
Consider the  infinite family   $\{{\cal G}_i\}$ of TVGs defined in Section \ref{sec:basic}.
Notice that
the duration of the fastest journey from $u$ to $v$  is 3 in    ${\cal G}_0$ while it is 2
in any other ${\cal G}_i$.
By Lemma \ref{newlemma},  {\sc Test}
 is undecidable even if both $n$ and $\Delta$ are known. 
 It follows that it is undecidable whether the fastest journey from $u$ to $v$ 
 has length 3 (in the case of ${\cal G}_0$)  or 2 (for all other ${\cal G}_i$).
 \end{proof}

The next question is what knowledge allows the problem to become feasible in \PG.
Observe that $\{ {\cal G}_i\} \setminus p$ is a finite set, hence Theorem \ref{lemma} does not apply.
Indeed, it has been shown in \cite{CFMS14} that, if  the   period $p$ is known, \fastest becomes feasible in \PG.

\begin{theorem}[from \cite{CFMS14}]
    \fastest is feasible in \PG with a known period $p$, and the corresponding broadcast tree can be reused as such in the subsequent periods.
  \label{th:fastest-P}
\end{theorem}

We can actually show that  \fastest is feasible in \PG with the weaker knowledge of  an upper bound  on the period.

\begin{theorem}
    \fastest is feasible in \PG with a known upper bound on $p$, and the broadcast tree can be reused for subsequent broadcasts.
\end{theorem}
 \begin{proof}
   The construction of fastest (broadcast) trees in~\cite{CFMS14} is based on the observation that achieving a fastest broadcast from some node comes to performing a {\em foremost} broadcast from that node at the time of its minimum {\em temporal eccentricity}, that is, when it takes the minimum time to reach all other nodes. The algorithm in~\cite{CFMS14} consists of learning this date over a complete period $p$, then build a foremost broadcast tree that starts at that date (modulo $p$). In case of multiplicity, any of the minimum values qualifies just as well. Having an upper bound $p^+$ on the period $p$ allows for the exact same technique. Indeed, any interval of time of length $p^+$ must contain an interval of length $p$, thus the time of minimum eccentricity will be detected. 
 \end{proof}
 
Note that in both cases (knowing $p$ or $p+$), the broadcast tree that is built remains necessarily optimal in the future, since in \PG the network schedule repeats forever. It can thus be {\em memorized} for subsequent broadcasts, {\it i.e.,} the solution is {\em reusable}.

\section{Computational Relationship}

On the basis of this paper results, we can prove   the validity of Equation~1 by showing the existence of
a strict computational hierarchy between $\PP(\RG_n)$, $\PP(\BG_\Delta)$, and $\PP(\PG_p)$.

\begin{theorem}
  $\PP(\RG_n)\subsetneq\PP(\BG_\Delta)\subsetneq\PP(\PG_p)$
  \label{th:strict-inclusion}
\end{theorem}
\begin{proof}
  The fact that $\PP(\RG_n)\subseteq\PP(\BG_\Delta)\subseteq\PP(\PG_p)$ was observed in Theorem~\ref{th:inclusion}. To make the left inclusion strict, one has to exhibit a problem $P$ such that $P \in \PP(\BG_\Delta)$ and $P \notin \PP(\RG_n)$. By Theorem~\ref{th:neg-shortest-P} and Theorem~\ref{th:shortest-B}, \shortest is one such example. The right inclusion is similarly proven strict, based on the fact that \fastest is in $\PP(\PG_p)$ (Theorem~\ref{th:fastest-P}) but it is not in 
  $\PP(\PG_{\Delta,n})$  and thus in $\PP(\BG_{\Delta})$  (Theorem~\ref{th:neg-fastest-n-delta}).
\end{proof}

With regards to the {\em feasibility} of the three   broadcast problems investigated here, the results established in the previous sections
indicate the existence of a strict hierarchy. 
Given a class ${\cal C}$ of TVGs, a knowledge $K$, and two problems $P_1, P_2 $,
 let $P_1 \preceq_{{\cal C}_K} P_2$  denote the  fact that for all ${\cal G}\in {\cal C}$, if  $P_2$ is feasible in ${\cal G}$
 with knowledge $K$,  so is $P_1$.

Given a (possibly infinite) set ${\cal L} = \{ {\cal C}^1_{K_1}, {\cal C}^2_{K_2} \ldots \}$ of classes of TVGs with given knowledge, let
$P_1 \preceq_{\cal L} P_2$    denote  the fact that  $P_1 \preceq_{{\cal C}_{K}} P_2$ for all $ {\cal C}_{K}  \in {\cal L}$. 
Let  $P_1 \prec_{\cal L} P_2$ denote the fact that    $P_1 \preceq_{\cal L} P_2$ and there exists at least one   ${\cal C}_K   \in {\cal L}$
such that  $P_1 \in \PP({\cal C}_K )$   but $P_2 \notin  \PP({\cal C}_K )$.
  
 Let ${\cal U} = \{ \RG_n, \BG_n, \BG_\Delta,\BG_{\{n,\Delta\}},  \PG_n,   \PG_\Delta,   \PG_{\{n,\Delta\}},  \PG_p \} $ be the universe of all the classes and knowledge considered in this paper.
 
 \begin{theorem}
$\foremost \prec_{\cal U} \shortest \prec_{\cal U} \fastest$
  \label{th:strict-inclusion2}
\end{theorem}
\begin{proof}
By Theorem~\ref{theo:foremost-R-n-known}, \foremost is feasible in $\RG_n$ (and {\it a fortiori} in $\BG_n$ and $\PG_n$, which are subsets of $\RG_n
$). By Theorem~\ref{th:inclusion}, it is also feasible in $\BG_\Delta$ (and in all the remaining combinations of ${\cal U}$, which are subsets of $\BG_\Delta$). Thus \foremost is feasible in ${\cal U}$ regardless of the class or knowledge considered. Now, by Theorem~\ref{th:neg-shortest-P}, \shortest is unfeasible in $\PG_n$ (among others). Thus, we have $\foremost \prec_{\cal U} \shortest$. Similarly, by Theorem~\ref{th:shortest-B}, \shortest is feasible in $\BG_\Delta$, and thus, among others, in $\PG_p$ (Theorem~\ref{th:inclusion}), whereas \fastest is only feasible in $\PG_p$ (Theorem~\ref{th:neg-fastest-n-delta}). It thus also holds that $\shortest \prec_{\cal U} \fastest$.
\end{proof}

With regards to {\em reusability}, the relationship between   the three   broadcast problems is drastically different.
For reusability, let  ${\leq}_{\cal L} $,  $<_{\cal L}$
 be the analogous of $\preceq_{\cal L}$,  $\prec_{\cal L}$ defined for feasibility.
Furthermore, let  $P_1 \equiv_{\cal L} P_2$ if   $P_1 \leq_{{\cal L}}  P_2$  and $P_2 \leq_{{\cal L}} P_1$.

 \begin{theorem}
$ \shortest <_{\cal U} \foremost \equiv_{\cal U} \fastest $
  \label{th:3}
 \end{theorem}
\begin{proof}
\shortest enables reusability in $\BG_\Delta$ (see the end of Section~\ref{sec:shortest-delta}) and thus in all the stronger contexts $\{\BG_{\{n,\Delta\}}$, $\PG_\Delta$,$\PG_{\{n,\Delta\}}$, $\PG_p\}$. On the other hand, \foremost, although feasible in all of ${\cal U}$'s contexts, enables reusability only in $\PG_p$ (Theorem~\ref{th:no-reuse-P-ndelta} and~\ref{th:reuse-P-p}). As for \fastest, it is only feasible (and reusable) in $\PG_p$ (Theorems~\ref{th:neg-fastest-n-delta} and~\ref{th:fastest-P}).
\end{proof}

Theorems~\ref{th:strict-inclusion2} and~\ref{th:3} suggest that
the difficulty of these problems is multi-dimensional, in that it depends on the aspect that is looked at (feasibility {\it vs.} reusability). Indeed, while \shortest is harder than \foremost in terms of feasibility, it is easier in terms of reusability. On the other hand, \fastest is (among) the hardest in both terms.

\section{Concluding Remarks}

In this paper we focused on three particular problems (shortest,
fastest, and foremost broadcast) in three classes of dynamic graphs
(recurrent, time-bounded recurrent, and periodic graphs) with
different types of applicable knowledge (size of the network, bound on
edge recurrence, period, upper bound on period). By comparing the
feasibility of these problems within each class depending on the
available knowledge, we have observed the impact that knowledge has on
feasibility and we have understood the computational relationship
between the classes in this context. This has in turn allowed us to
observe the relative ``difficulty" of the problems under
investigation.

Among other things our  results show, for example, the special importance of  periodic dynamic graphs with known period, the only combination of class and knowledge (in the universe considered here) where Fastest broadcast is feasible. 
It also stresses the inherent difference between reusability of
Foremost broadcast  (which is the ``easiest" problem to solve but is reusable only in periodic graphs with known period), 
and Fastest and Shortest on the other, which can be reused whenever they can be solved.
Another interesting observation  stemming from our results is the intrinsic limitation  of knowing only the number of nodes,  in which case, regardless of the class of graphs considered, only Foremost broadcast can be performed (without  being able to reuse it).

This study is a first step toward an understanding of computability in dynamic graphs and it
 opens the door to  more general investigations on the  computability power of different classes and their relationship with  knowledge available to the nodes.

\paragraph{Acknowledgments:}
The authors would like to thank the anonymous referees. Their questions and comments have led to a stronger and clearer paper.


\end{document}